\RequirePackage{fix-cm}
\documentclass{svjour3}                     % onecolumn (standard format)
\smartqed  % flush right qed marks, e.g. at end of proof
\usepackage{graphicx}
\usepackage{amssymb}
\usepackage{amsmath}
\usepackage{algorithm2e}

\journalname{Designs, Codes and Cryptography}
\begin{document}

\title{The weight distributions of some cyclic codes with three or four nonzeros over $\mathbb{F}_3$
\thanks{This work is supported in part by the National Key Basic Research and Development Plan of China under Grant 2012CB316100, and
the National Natural Science Foundation of China under Grants 61271222, 60972033.}
}

\titlerunning{cyclic codes with three or four nonzeros}

\author{Xiaogang Liu \and Yuan Luo}


\institute{Xiaogang Liu
\at Department of Computer Science and Engineering, Shanghai Jiao Tong University, China
 \\\email{liuxg0201@163.com}
\and
Yuan Luo (Corresponding author)
\at Department of Computer Science and Engineering, Shanghai Jiao Tong University, China
\\Addr.: 800 Dongchuan Road, Min Hang District, Shanghai 200240, China.
\\Tel.:+86 21 34205477 \email{yuanluo@sjtu.edu.cn}
}

\date{Received: date / Accepted: date}
\maketitle

\begin{abstract}
Because of efficient encoding and decoding algorithms, cyclic codes are an important family of linear block codes, and have applications in communication and storage systems.
However, their weight distributions are known only for a few cases mainly on the codes with one or two nonzeros. In this paper, the weight distributions of two classes of cyclic codes with three or four nonzeros are determined.

\end{abstract}

\section{Introduction} \label{Sec1}

%For a cyclic code $\mathcal{C}$ with length $n$ over a finite field $\mathbb{F}_p$ where $p$ is a prime, let $A_i$ be the number of codewords in $\mathcal{C}$ with Hamming weight $i$. The weight distribution $\{A_0,A_1,\ldots,A_n\}$ is an important research object in coding theory. If $\mathcal{C}$ is irreducible, which means that the parity check polynomial of $\mathcal{C}$ is irreducible in $\mathbb{F}_p[x]$, the weight of each codeword can be expressed by Gaussian sums so that the weight distribution of $\mathcal{C}$ can be determined if the corresponding Gaussian sums(or their certain combinations) can be calculated explicitly (see \cite{M01,YC01}, and the references therein). As for the relationship between the weight distribution of cyclic codes and the rational points of certain curves, see \cite{R001}.

 An $[n,k,d;p]$ code is a $k$-dimensional subspace of $\mathbb{F}_p^n$ with minimum(Hamming) distance $d$. Let $A_i$ denote the number of codewords with Hamming weight $i$ in a code $\mathcal{C}$ of length $n$. The {\bf weight enumerator} of $\mathcal{C}$ is defined by
\[
1+A_1x+A_2x^2+\cdots+A_nx^n.
\]
The sequence $(1,A_1,\cdots,A_n)$ is called the {\bf weight distribution} of the code, which is an important parameter of a linear block code.
% The study of the weight distribution of a linear code is important both in theory and applications due to the following:
%\begin{enumerate}
%\renewcommand{\labelenumi}{\arabic{enumi})}
%\item
In fact, the minimum distance $d$ determines the error correcting capability of the code $\mathcal{C}$.
%\item
Furthermore, under some algorithms, we can compute the error probability of error detection and correction.%The weight distribution of a code allows the computation of the error probability of error detection and correction with respect to some error detection and error correction algorithms.
%\end{enumerate}

An $[n,k,d;p]$ linear code $\mathcal{C}$ is called cyclic if $(c_0,c_1,\cdots,c_{n-1})\in \mathcal{C}$ implies that $(c_{n-1},c_0,c_1,\cdots,c_{n-2})\in \mathcal{C}$ $(\mbox{gcd}(n,p)=1)$.
%A cyclic code $\mathcal{C}$ of length $n$ over $\mathbb{F}_p$ is also called a $p$-ary cyclic code of length $n$.
 Using the vector space isomorphism from $\mathbb{F}_p^m$ to the principal ideal ring $R_n:=\mathbb{F}_p[x]/(x^n-1)$
 \[
 (a_0,a_1,\ldots,a_{n-1})\longrightarrow
a_0+a_1x+\cdots a_{n-1}x^{n-1},
\]
$\mathcal{C}$ is an ideal. The generator $g(x)$ of this ideal is called the generating polynomial of $\mathcal{C}$, which satisfies that $g(x)|(x^n-1)$. When the ideal is minimal, the code $\mathcal{C}$ is called an irreducible cyclic code.
  % A minimal ideal in $R_n$ is called an irreducible cyclic code of length $n$ over $\mathbb{F}_p$.
   For any $v=(c_0,c_1,\cdots,c_{n-1})\in \mathcal{C}$, the weight of $v$ is $wt(v)=\#\{c_i\not=0,i=0,1,\ldots,n-1\}$.

Many authors have studied how to determine the weight distributions of cyclic codes. MacWilliams and Seery \cite{MS01} gave a procedure for binary cyclic codes, but it can be implemented only on a powerful computer. The problem of computing weight distributions is connected to the evaluation of certain exponential sums McEliece, Rumsey \cite{MR001} and Van Der Vlugt \cite{M01}, which are generally hard to determine explicitly. %To evaluate these sums in some special cases, certain algorithms were given by Baumert and McEliece.., Moisio and Vaanaanen..Fitzgerald and Yucas, etc., using various techniques. %In.., Augot used the theory of Grobner basis for a certain system of algebraic equations to give information about the minimum weight codewords.
In \cite{R001}, Schoof studied its relation with the rational points of certain curves. For the weight distributions of the cyclic codes with three nonzeros, please refer to Feng, Luo \cite{JH001} and Zeng, Hu, etc. \cite{JHK001}. % \cite{FL001,ZH001}.
 The related problems in the binary cases with two nonzeros, were analyzed in Johansen, Helleseth, Kholosha \cite{JH001,JHK001}.

Assume that $p$ is an odd prime, $q=p^m$ for a positive odd integer $m$. % and $3\nmid m$.
 Let $\pi$ be a primitive element of $\mathbb{F}_q$. This paper determines the weight distributions of cyclic codes $\mathcal{C}_1$ and $\mathcal{C}_2$ over $\mathbb{F}_3$ with nonzeros $\pi^{-2},\pi^{-4},\pi^{-10}$ and $\pi^{-1},\pi^{-2},\pi^{-4},\pi^{-10}$ respectively, the weight distributions of which are verified by two examples using matlab. Note that the length of the cyclic codes is $l=q-1=3^m-1$ and $3\nmid m$.

In the following, Section \ref{Sec2} presents the basic notations and results about cyclic codes. Section \ref{Sec3} focus on the class of cyclic code $\mathcal{C}_1$. Section \ref{Sec4} is on the class of cyclic code $\mathcal{C}_2$. Final conclusion is in Section \ref{Sec5}.

\section{Preliminaries} \label{Sec2}

%\subsection{Relevant results from finite fields} \label{Sec3.2}
In this section, relevant knowledge from finite fields is presented first for our study of cyclic codes in Section \ref{Sec2.1}. Then some results about the calculations of exponential sums are presented in Section \ref{Sec2.2}. Section \ref{Sec2.3} concerns the sizes of cyclotomic cosets and the ranks of certain quadratic forms.%Then some results about the dimensions of particular quadratic forms and the calculations of certain exponential sums are given in Section \ref{Sec3.2.2} and Section \ref{Sec3.2.3} respectively for the weights of codewords in $\mathcal{C}$. Note that Lemmas \ref{FL002}, \ref{DH001} and \ref{G01} are known results.

\subsection{Finite fields and cyclic codes} \label{Sec2.1}

Here, some known properties about the codeword weight are listed, and the mathematical tools exponential sums and quadratic forms are introduced.
For more researches about cyclic codes, refer to \cite{D001,DY001,FY001,M001} for the irreducible case, and \cite{FL001,LF001,V001,YC01} for the reducible case.
%\noindent{\sl Weight of codeword in cyclic code:}

%  For an odd prime $p$, let the cyclic code $\mathcal{C}$ over $\mathbb{F}_p$ be of length $l=q-1=p^m-1$ with  parity check polynomial
%\[
%h(x)=h_1(x)\cdots h_{\iota}(x) \ \ \ (\iota \geq 1),
%\]
%where $h_\lambda(x) (1\leq \lambda \leq \iota)$ are distinct irreducible polynomials in $\mathbb{F}_p[x]$ with the degrees $e_\lambda (1\leq \lambda\leq \iota)$, then $k=\mbox{dim}_{\mathbb{F}_p}\mathcal{C} = \sum\limits_{\lambda=1}^{\iota}e_\lambda$. Let $\pi$ be a primitive element of $\mathbb{F}_q$ and $\pi^{-s_\lambda}$ be a zero of $h_\lambda(x), 1\leq s_\lambda \leq q-2 (1\leq \lambda \leq \iota)$. Then the codewords in $\mathcal{C}$ can be expressed by
%\begin{equation}\label{CW001}
%c(\alpha_1,\ldots,\alpha_{\iota}) = (c_0,c_1,\ldots,c_{l-1}) \ \ (\alpha_1,\ldots,\alpha_{\iota}\in \mathbb{F}_q),
%\end{equation}
%where $c_i = \sum\limits_{\lambda =1}^{\iota}\mbox{Tr}(\alpha_{\lambda}\pi^{is_{\lambda}}) (0\leq i\leq l-1)$ and $\mbox{Tr}:\mathbb{F}_q\rightarrow \mathbb{F}_p$ is the trace mapping from $\mathbb{F}_q$ to $\mathbb{F}_p$. Therefore the Hamming weight of the codeword $c=c(\alpha_1,\ldots,\alpha_{\iota})$ is:

 For an odd prime $p$ and a positive integer $m$, let the cyclic code $\mathcal{C}$ over $\mathbb{F}_p$ be of length $l=q-1=p^m-1$ with non-conjugate nonzeros $\pi^{-s_\lambda}$, where $1\leq s_\lambda \leq q-2 (1\leq \lambda \leq \iota)$ and $\pi$ is a primitive element of $\mathbb{F}_q$. Then the codewords in $\mathcal{C}$ can be expressed by
\begin{equation}\label{CW001}
c(\alpha_1,\ldots,\alpha_{\iota}) = (c_0,c_1,\ldots,c_{l-1}) \ \ (\alpha_1,\ldots,\alpha_{\iota}\in \mathbb{F}_q),
\end{equation}
where $c_i = \sum\limits_{\lambda =1}^{\iota}\mbox{Tr}(\alpha_{\lambda}\pi^{is_{\lambda}}) (0\leq i\leq l-1)$ and $\mbox{Tr}:\mathbb{F}_q\rightarrow \mathbb{F}_p$ is the trace mapping from $\mathbb{F}_q$ to $\mathbb{F}_p$. Therefore the Hamming weight of the codeword $c=c(\alpha_1,\ldots,\alpha_{\iota})$ is:

\begin{equation}\label{C01}
\begin{array}{ll}
w_H(c)&=\# \{i|0\leq i \leq l-1, c_i\not= 0\}\\
     % &= l-\# \{i|0\leq i \leq l-1, c_i =0\}\\
     % &=l-{1\over p}\sum\limits_{i=0}^{l-1}\sum\limits_{a=0}^{p-1} {\zeta_p^{a\cdot \mbox{Tr}(\sum_{\lambda = 1}^{\iota}\alpha_{\lambda}\pi^{is_{\lambda}})}}\\
      &=l-{l\over p}-{1\over p}\sum\limits_{a=1}^{p-1}\sum\limits_{x\in \mathbb{F}_q^*}{\zeta}_p^{\mbox{Tr}(af(x))}\\
      %&=l-{l\over p}+{{p-1}\over p}-{1\over p}\sum\limits_{a=1}^{p-1}S(a\alpha_1,\ldots,a\alpha_{\iota})\\
      &=p^{m-1}(p-1)-{1\over p}\sum\limits_{a=1}^{p-1}S(a\alpha_1,\ldots,a\alpha_{\iota})\\
      &=p^{m-1}(p-1)-{1\over p}R(\alpha_1,\ldots,\alpha_{\iota})
\end{array}
\end{equation}
where  ${\zeta_p}=e^{{2\pi i}\over p}$ ($i$ is imaginary unit), $f(x)=\alpha_1x^{s_1}+\alpha_2x^{s_2}+\cdots +\alpha_{\iota}x^{s_{\iota}}\in \mathbb{F}_q[x], \mathbb{F}_q^*=\mathbb{F}_q\backslash \{0\}$,
\begin{equation}\label{ES0001}
S(\alpha_1,\ldots,\alpha_{\iota}) = \sum\limits_{x\in \mathbb{F}_q}{\zeta}_p^{\mbox{Tr}\left(\alpha_1x^{s_1}+\cdots +\alpha_{\iota}x^{s_{\iota}}\right)},
\end{equation}
and $R(\alpha_1,\ldots,\alpha_{\iota})=\sum\limits_{a=1}^{p-1}S(a\alpha_1,\ldots,a\alpha_{\iota})$.

%\begin{remark}
%There may not be a one-to-one correspondence between the codewords of $\mathcal{C}$ and equation (\ref{CW001}).
%\end{remark}

%\noindent {\sl Quadratic forms:}

% Fix a basis $v_1,\ldots,v_m$ of $\mathbb{F}_q$ over $\mathbb{F}_p$ where $q=p^m$, each $x\in \mathbb{F}_q$ can be uniquely expressed as
%\[
%x=x_1v_1+\ldots +x_mv_m \ \ \ (x_i \in \mathbb{F}_p).
%\]
%There is the following $\mathbb{F}_p$-linear isomorphism:
%\[
%\mathbb{F}_q \overset{\sim}{\rightarrow} \mathbb{F}_p^m, \ \ x=x_1v_1+\cdots+x_mv_m \mapsto X=(x_1,\ldots,x_m).
%\]
%With this isomorphism, a function $f:\mathbb{F}_q \rightarrow \mathbb{F}_q$ induces a function $F:\mathbb{F}_p^m \rightarrow \mathbb{F}_p$ where for %$X=(x_1,\ldots,x_m)\in \mathbb{F}_p^m, F(X)=\mbox{Tr}(f(x))$.

For general functions of the form $
f_{\alpha,\ldots,\gamma}(x) = \alpha x^{p^i+1} + \cdots + \gamma x^{p^j+1}$
 where $0\leq i,\ldots, j \leq  \lfloor{{m}\over 2}\rfloor$, there are quadratic forms
 \begin{equation}\label{QF01}
 F_{\alpha,\ldots,\gamma}(X)
 \end{equation}
  and corresponding symmetric matrices
  \begin{equation}\label{QF02}
  H_{\alpha,\ldots,\gamma}
  \end{equation}
  satisfying that $F_{\alpha,\ldots,\gamma}(X)=XH_{\alpha,\ldots,\gamma}X^T=\mbox{Tr}(f_{\alpha,\ldots,\gamma}(x)).$

%  \hspace{0.5cm}

%\noindent{\sl Symmetric matrix:}

%For an $m\times m$ symmetric matrix $H$ over $F_p$ and $r= \mbox{rank}H_{\alpha,\ldots,\gamma}$,
It is known that there exists $M_{\alpha,\ldots,\gamma} \in \mbox{GL}_m(\mathbb{F}_p)$ such that
 \[
H_{\alpha,\ldots,\gamma}'=M_{\alpha,\ldots,\gamma}H_{\alpha,\ldots,\gamma}M_{\alpha,\ldots,\gamma}^T=\mbox{diag}(a_1,\ldots,a_{r_{\alpha,\ldots,\gamma}},0,\ldots,0),
\]
where $a_i\in \mathbb{F}_p^* (1\leq i\leq r_{\alpha,\ldots,\gamma})$ and $r_{\alpha,\ldots,\gamma}= \mbox{rank}H_{\alpha,\ldots,\gamma}$. Let $\Delta =a_1\cdots a_{r_{\alpha,\ldots,\gamma}}$ (set $\Delta =1$ for $r_{\alpha,\ldots,\gamma}=0$), and
\begin{equation}\label{D01}
\left({\Delta \over p}\right)
\end{equation}
denotes the Legendre symbol.
The following result is about the exponential sum corresponding to the symmetric matrix $H_{\alpha,\ldots,\gamma}$ \cite{FL001}, see also \cite{LH001}.

\begin{lemma}(Lemma 1, \cite{FL001})\label{FL002}
\begin{enumerate}
\renewcommand{\labelenumi}{$($\mbox{\roman{enumi}}$)$}
\item
For the quadratic form $F(X) = XHX^T$, %defined in (\ref{QF001}),
\[
\sum\limits_{X\in \mathbb{F}_p^m}\zeta_p^{F(X)} =
\begin{cases}
\left({\Delta \over p}\right)p^{m-r/2} & \mbox{if} \ \ p\equiv 1 \ (\mbox{mod} \ 4),\\
i^r\left({\Delta \over p}\right)p^{m-r/2} & \mbox{if} \ \ p \equiv 3 \ (\mbox{mod} \ 4).
\end{cases}
\]
\item
For $A = (a_1,\ldots,a_m)\in \mathbb{F}_p^m$, if \ $2YH+A=0$ has solution $Y=B \in \mathbb{F}_p^m$, then
\begin{equation}
\sum_{X\in \mathbb{F}_p^m}{\zeta}_p^{F(X)+AX^T} = {\zeta}_{p}^c\sum_{X\in \mathbb{F}_p^m}{\zeta}_p^{F(X)} \ \mbox{where} \ \ c={1\over 2}AB^T\in \mathbb{F}_p.
\end{equation}
Otherwise $\sum\limits_{X\in \mathbb{F}_p^m}{\zeta}_p^{F(X)+AX^T} = 0$.
\end{enumerate}
\end{lemma}

\subsection{Results about exponential sums}\label{Sec2.2}

In this subsection Lemma \ref{ES02} and Remark \ref{R0022} are from \cite{LL001}, also refer to \cite{FL001} for the calculations of exponential sums that will be needed in the sequel.
\begin{lemma}\label{ES02}
For the quadratic form $F_{\alpha,\ldots,\gamma}(X) = XH_{\alpha,\ldots,\gamma}X^T$ corresponding to $f_{\alpha,\ldots,\gamma}(x)$, see (\ref{QF01})
\begin{enumerate}
\renewcommand{\labelenumi}{$($\mbox{\roman{enumi}}$)$}
\item
if the rank $r_{\alpha,\ldots,\gamma}$ of the symmetric matrix $H_{\alpha,\ldots,\gamma}$ is even, which means that $S(\alpha,\ldots,\gamma) = \varepsilon p^{m-{r_{\alpha,\ldots,\gamma}\over 2}}$, then
\[
R(\alpha,\ldots,\gamma) =\varepsilon (p-1) p^{m-{r_{\alpha,\ldots,\gamma}\over 2}}; %\sum\limits_{a=1}^{p-1}S(a\alpha,\ldots,a\beta)
\]
\item
if the rank $r_{\alpha,\ldots,\gamma}$ of the symmetric matrix $H_{\alpha,\ldots,\gamma}$ is odd, which means that $S(\alpha,\ldots,\gamma) = \varepsilon \sqrt{p^*}p^{m-{{r_{\alpha,\ldots,\gamma}+1}\over 2}}$, then
\[
R(\alpha,\ldots,\gamma) = 0
\]
\end{enumerate}
where $ \varepsilon = \pm 1$ and $p^*=\left({{-1} \over p}\right)p$.
\end{lemma}

\begin{lemma}\label{RO01}
 Let $F_{\alpha,\ldots,\gamma}(X) = XH_{\alpha,\ldots,\gamma}X^T$ be the quadratic form corresponding to $f_{\alpha,\ldots,\gamma}(x)$, see (\ref{QF01}).
If the rank $r_{\alpha,\ldots,\gamma}$ of the symmetric matrix $H_{\alpha,\ldots,\gamma}$ is odd, then the number of quadratic forms with exponential sum $\sqrt{p^*}p^{m-{{r_{\alpha,\ldots,\gamma}+1}\over 2}}$ equals the number of quadratic forms with exponential sum $-\sqrt{p^*}p^{m-{{r_{\alpha,\ldots,\gamma}+1}\over 2}}$ where $p^*=\left({{-1} \over p}\right)p$.
\end{lemma}

\begin{proof}
%For $f_{\alpha,\ldots,\gamma}(x)$ with corresponding rank $r_{\alpha,\ldots,\gamma}$ and symmetric matrix $H_{\alpha,\ldots,\gamma}$.
 Choose a quadratic nonresidue $a\in \mathbb{F}_p^*$, then the symmetric matrix corresponding to $af_{\alpha,\ldots,\gamma}(x)$ is $aH_{\alpha,\ldots,\gamma}$ which also has rank $r_{\alpha,\ldots,\gamma}$, and
 \[
 M_{\alpha,\ldots,\gamma}\left(aH_{\alpha,\ldots,\gamma}\right)M_{\alpha,\ldots,\gamma}^T=\mbox{diag}(aa_1,\ldots,aa_{r_{\alpha,\ldots,\gamma}},0,\ldots,0).
 \] Since $r_{\alpha,\ldots,\gamma}$ is odd,
\[
\left({\Delta \cdot a^{r_{\alpha,\ldots,\gamma}} \over p}\right) =\left({\Delta \cdot a \over p}\right)=\left({\Delta \over p}\right)\cdot \left({a \over p}\right)=-\left({\Delta \over p}\right)
\]
where $\Delta =a_1\cdots a_{r_{\alpha,\ldots,\gamma}}$.
The result follows from Lemma \ref{FL002} and the statement above it.  \qed
\end{proof}

\begin{remark}\label{R0022}
For the exponential sum $S(\alpha,\ldots,\gamma)$ corresponding to $f_{\alpha,\ldots,\gamma}(x)=\alpha x^{p^i+1}+\cdots+\gamma x^{p^j+1}$ with quadratic form $F_{\alpha,\ldots,\gamma}(X)$ and symmetric matrix $H_{\alpha,\ldots,\gamma}$ (equation (\ref{QF01})), consider $S'(\alpha,\ldots,\gamma,\delta)$ with respect to
\begin{equation}\label{0005}
f'_{\alpha,\ldots,\gamma,\delta}(x)=f_{\alpha,\ldots,\gamma}(x)+\delta x,
\end{equation}
 and $R'(\alpha,\ldots,\gamma,\delta)=\sum\limits_{a=1}^{p-1}S'(a\alpha,\ldots,a\gamma,a\delta)$ (equation (\ref{C01})). From Lemma \ref{FL002}, there are four cases to be considered where the first two equations are for the case with symmetric matrices $H$ of even rank and the last two equations for the case of odd rank. %Then%Now, by Lemma \ref{G01} we can calculate
\begin{itemize}
\renewcommand{\labelitemi}{\labelitemiii}
\item
If $S'(\alpha,\ldots,\gamma,\delta)=\varepsilon p^{r'}$, then $R'(\alpha,\ldots,\gamma,\delta)= %\sum\limits_{a=1}^{p-1}\sigma_a(S'(\alpha,\ldots,\gamma,\delta)) = \sum\limits_{a=1}^{p-1}\varepsilon p^{r'} =
\varepsilon (p-1)p^{r'}; $
\item
If $S'(\alpha,\ldots,\gamma,\delta)=\varepsilon \zeta_p^c p^{r'}$, then $R'(\alpha,\ldots,\gamma,\delta)=% \sum\limits_{a=1}^{p-1}\sigma_a(S'(\alpha,\ldots,\gamma,\delta))=\varepsilon p^{r'}\sum\limits_{a=1}^{p-1}\zeta_p^{ac} =
-\varepsilon p^{r'};$
\item
If $S'(\alpha,\ldots,\gamma,\delta)=\varepsilon \sqrt{p^*}p^{r'}$, then $R'(\alpha,\ldots,\gamma,\delta)= %\sum\limits_{a=1}^{p-1}\sigma_a(S'(\alpha,\ldots,\gamma,\delta)) = \varepsilon \sqrt{p^*}p^{r'}\sum\limits_{a=1}^{p-1}\left(a\over p\right) =
0;$
\item
If $S'(\alpha,\ldots,\gamma,\delta)=\varepsilon \zeta_p^c \sqrt{p^*} p^{r'}$, then $R'(\alpha,\ldots,\gamma,\delta)= %\sum\limits_{a=1}^{p-1}\sigma_a(S'(\alpha,\ldots,\gamma,\delta))=\varepsilon \sqrt{p^*}p^{r'}\sum\limits_{a=1}^{p-1}\left(a\over p\right)\zeta_p^{ac} =
\varepsilon \left({-c}\over p\right)p^{{r'}+{1}}$.
\end{itemize}
In the above, $r'$ is a positive integers, $c \in \mathbb{F}_p^* $, $p^*=\left({-1}\over p\right)p$ and $ \varepsilon =\pm 1$.
\end{remark}

\subsection{Cyclotomic cosets and the ranks of certain quadratic forms}\label{Sec2.3}

The cyclotomic coset containing $s$ is defined to be
\begin{equation}\label{CC01}
\mathcal{D}_s=\{s,sp,sp^2,\ldots,sp^{m_s-1}\}
\end{equation}
where $m_s$ is the smallest positive integer such that $p^{m_s}\cdot s \equiv s \ (\mbox{mod} \ p^m-1)$.

In the following, Lemma \ref{GCD002} and Lemma \ref{RQ002} are from \cite{LL001}, also refer to \cite{CH001} for the binary case of Lemma \ref{GCD002}.
\begin{lemma}\label{GCD002}
%\begin{enumerate}
%\item
If $m=2t+1$ is odd, then for $l_i=1+p^i$, the cyclotomic coset $\mathcal{D}_{l_i}$ has size
\[
|\mathcal{D}_{l_i}|=m, \ \ \ 0\leq i \leq t.
\]
%\item
If $m=2t+2$ is even, then for $l_i=1+p^i$, the cyclotomic coset $\mathcal{D}_{l_i}$ has size
\[
|\mathcal{D}_{l_i}| =
\begin{cases}
m, & 0\leq i \leq t \\
m/2, & i=t+1.
\end{cases}
\]
%where $t\ge 1$ is a positive integer.
%\end{enumerate}
\end{lemma}

For $f_d'(x) =\alpha_0 x^2+\alpha_1x^{p+1}+\cdots+\alpha_dx^{p^d+1}$ with corresponding quadratic form $F_d'(X)=\mbox{Tr}(f_d'(x))=XH_d'X^T$ where $(\alpha_0,\alpha_1,\ldots,\alpha_d)\in \mathbb{F}_q^{d+1}\backslash\{(0,0,\ldots,0)\}$, the following result is about its rank.
\begin{lemma}\label{RQ002}
Let $m$ be a positive integer, $0\leq d\leq \lfloor{m\over 2}\rfloor$. The rank $r_d'$ of the symmetric matrix $H_d'$ satisfies $r_d'\geq m-2d$.
\end{lemma}

The following corollary is a special case of Lemma \ref{RQ002}.% From the above lemma, there is the following result.
\begin{corollary}\label{RQ02}
The rank $r_2'$ of the symmetric matrix $H_2'$ corresponding to $f_2'(x) =\alpha_0 x^2+\alpha_1x^{p+1}+\alpha_2x^{p^2+1}$ has five possible values:
\[
m,m-1,m-2,m-3,m-4.
\]
\end{corollary}

\section{The cyclic code $\mathcal{C}_1$} \label{Sec3}

This section investigates the weight distribution of the cyclic code $\mathcal{C}_1$ over $\mathbb{F}_3$ with length $l=3^m-1$ and nonzeros $\pi^{-2},\pi^{-4}$ and $\pi^{-10}$, where $\pi$ is a primitive element of the finite field $\mathbb{F}_{3^m}$ for an odd integer $m$ satisfying $3\nmid m$.

First, Lemma \ref{SE01} and Lemma \ref{SE02} are stated about the number of solutions of quadratic equations over finite field. Secondly moments of exponential sum $S(\alpha,\beta,\gamma)$ are calculated in Section \ref{Sec3.1} and Section \ref{Sec3.2}, which provide four equations and one equation for the weight distributions respectively. Finally, some relevant results about quadratic forms are presented in Section \ref{Sec3.3} which provides another two equations using association schemes, and main results are provided by using the seven equations in Theorem \ref{C002} of Section \ref{Sec3.4}.

\begin{definition}
For any finite field $\mathbb{F}_q$ the integer-valued function $\upsilon$ on $\mathbb{F}_q$ is defined by $\upsilon (b)=-1$ for $b\in \mathbb{F}_q^*$ and $\upsilon (0)=q-1$.
\end{definition}

\begin{lemma}(Theorem 6.26., \cite{LH001}) \label{SE01}
Let $f$ be a nondegenerate quadratic form over $\mathbb{F}_q$, $q$ odd, in an even number $n$ of indeterminates. Then for $b\in \mathbb{F}_q$ the number of solutions of the equation $f(x_1,\ldots,x_n)=b$ in $\mathbb{F}_q^n$ is
\[
q^{n-1}+\upsilon (b)q^{(n-2)/ 2} \eta \left((-1)^{n/ 2}\Delta\right)
\]
where $\eta$ is the quadratic character of $\mathbb{F}_q$ and $\Delta =\mbox{det}(f)$.
\end{lemma}

\begin{lemma}(Theorem 6.27., \cite{LH001}) \label{SE02}
Let $f$ be a nondegenerate quadratic form over $\mathbb{F}_q$, $q$ odd, in an odd number $n$ of indeterminates. Then for $b\in \mathbb{F}_q$ the number of solutions of the equation $f(x_1,\ldots,x_n)=b$ in $\mathbb{F}_q^n$ is
\[
q^{n-1}+q^{(n-1)/ 2} \eta \left((-1)^{(n-1)/ 2}b\Delta\right)
\]
where $\eta$ is the quadratic character of $\mathbb{F}_q$ and $\Delta =\mbox{det}(f)$.
\end{lemma}

\subsection{Moments of the exponential sum $S(\alpha,\beta,\gamma)$}\label{Sec3.1}

%We study the cyclic code $\mathcal{C}_1$ over the finite field $\mathbb{F}_p$ of length $l=p^m-1$, whose nonzeros include $\pi^{-2},\pi^{-(p^k+1)}$ and $\pi^{-(p^{2k}+1)}$ where $k=1$. And without particular specification, we assume that $p=3,3\nmid m$, $m$ is an odd integer. %, and $\mathcal{C}$ is . %As in ..., there are weight formulas and results about quadratic forms..... $3$ is not a divisor of $m$.

For an odd prime $p$, this subsection calculates the first three moments of the exponential sum $S(\alpha,\beta,\gamma)$ (equation \ref{ES0001}).

\begin{lemma}\label{ES001}
Let $p$ be an odd prime satisfying $p\equiv 3 \ \mbox{mod}\ 4$, and $q=p^m$ where $m$ is an odd integer with property $3\nmid m$.
 %With notations as in equation (\ref{ES0001}),
 Then there are the following results about the exponential sum $S(\alpha,\beta,\gamma)$ (equation \ref{ES0001}) corresponding to $f_2'(x) =\alpha x^2+\beta x^{p+1}+\gamma x^{p^2+1}$ %where $(\alpha,\beta,\gamma)\in \mathbb{F}_q^3$ %\backslash\{0,0,0\}$.
\begin{enumerate}
\renewcommand{\labelenumi}{$($\mbox{\roman{enumi}}$)$}
\item
$\sum\limits_{\alpha,\beta,\gamma\in \mathbb{F}_q} S(\alpha,\beta,\gamma)=p^{3m}$
\item
$\sum\limits_{\alpha,\beta,\gamma\in \mathbb{F}_q} S(\alpha,\beta,\gamma)^2=p^{3m}$
\item
$\sum\limits_{\alpha,\beta,\gamma\in \mathbb{F}_q} S(\alpha,\beta,\gamma)^3=\left((p+1)(p^m-1)+1\right)p^{3m}$.
\end{enumerate}
\end{lemma}

\begin{proof}
From definition, changing the order of summations,
(i) can be calculated as follows
\[
\begin{array}{ll}
\sum\limits_{\alpha,\beta,\gamma\in \mathbb{F}_q} S(\alpha,\beta,\gamma)
&=\sum\limits_{\alpha,\beta,\gamma\in \mathbb{F}_q}\sum\limits_{x\in\mathbb{F}_q}\zeta_p^{\mbox{Tr}\left(\alpha x^2+\beta x^{p+1}+\gamma x^{p^2+1}\right)} \\
 &=\sum\limits_{x\in \mathbb{F}_q}\sum\limits_{\alpha \in \mathbb{F}_q}\zeta_p^{\mbox{Tr}\left(\alpha x^2\right)}\sum\limits_{\beta \in \mathbb{F}_q}\zeta_p^{\mbox{Tr}\left(\beta x^{p+1}\right)}\sum\limits_{\gamma \in \mathbb{F}_q}\zeta_p^{\mbox{Tr}\left(\gamma x^{p^2+1}\right)}   \\
 &=\sum\limits_{\stackrel{\alpha \in \mathbb{F}_q}{x=0}}\zeta_p^{\mbox{Tr}\left(\alpha x^2\right)}\sum\limits_{\stackrel{\beta \in \mathbb{F}_q}{x=0}}\zeta_p^{\mbox{Tr}\left(\beta x^{p+1}\right)}\sum\limits_{\stackrel{\gamma \in \mathbb{F}_q}{x=0}}\zeta_p^{\mbox{Tr}\left(\gamma x^{p^2+1}\right)} \\
 &=q^3=p^{3m}.
 \end{array}
\]

Equation (ii) can also be calculated in this way
\[
\begin{array}{lll}
&&\sum\limits_{\alpha,\beta,\gamma\in \mathbb{F}_q} S(\alpha,\beta,\gamma)^2 \\
&=&\sum\limits_{x,y\in \mathbb{F}_q}\sum\limits_{\alpha \in \mathbb{F}_q}\zeta_p^{\mbox{Tr}\left(\alpha\left(x^{2}+y^{2}\right)\right)}\sum\limits_{\beta \in \mathbb{F}_q}\zeta_p^{\mbox{Tr}\left(\beta\left(x^{p+1}+y^{p+1}\right)\right)}\sum\limits_{\gamma \in \mathbb{F}_q}\zeta_p^{\mbox{Tr}\left(\gamma\left(x^{p^2+1}+y^{p^2+1}\right)\right)}\\
&=&M_2\cdot p^{3m}
\end{array}
\]
where $M_2$ is the number of solutions to the equation system
\[
\left\{
 \begin{array}{ll}
 x^2+y^2 &=0\\
 x^{p+1}+y^{p+1} &=0\\
 x^{p^2+1}+y^{p^2+1} &=0.
 \end{array}
\right.
\]
Since it is assumed that $p\equiv 3 \ \mbox{mod}\ 4$ and $m$ is an odd integer, there is not an element $x_0 \in \mathbb{F}_q$ satisfying $x_0^2=-1$. The only solution to above system is $x=y=0$, that is $M_2=1$.

As to (iii), we have
\begin{equation}\label{CM02}
\sum\limits_{\alpha,\beta,\gamma\in \mathbb{F}_q} S(\alpha,\beta,\gamma)^3=M_3\cdot p^{3m}
\end{equation}
where
\[
\begin{array}{lll}
M_3&=\#\{(x,y,z)\in \mathbb{F}_q^3|&x^2+y^2+z^2=0,\\
    &       &x^{p+1}+y^{p+1}+z^{p+1}=0,x^{p^2+1}+y^{p^2+1}+z^{p^2+1}=0\}\\
 &=M_2+T_3\cdot (q-1),
\end{array}
\]
and $T_3$ is the number of solutions of
\begin{equation}\label{ES01}
\left\{
\begin{array}{ll}
x^2+y^2+1 &=0 \\
x^{p+1}+y^{p+1}+1 &=0 \\
x^{p^2+1}+y^{p^2+1}+1 &=0.
\end{array}
\right.
\end{equation}
To study equation system (\ref{ES01}), consider the last two equations. Canceling $y$ there is
\[
(x^{p+1}+1)^{p^2+1}=(x^{p^2+1}+1)^{p+1},
\]
after simplification, it becomes
\begin{equation}\label{ES0002}
(x^{p^2}-x^p)(x^{p^3}-x)=(x^{p}-x)^p(x^{p^3}-x)=0.
\end{equation}
Since $3$ is not a divisor of $m$, from (\ref{ES0002}) it can be checked that $x\in \mathbb{F}_p$. In the same way, it implies that $y\in \mathbb{F}_p$. Since $a^p=a$ for any $a\in \mathbb{F}_p$, we only need to consider the first one of system (\ref{ES01}). In case of Lemma \ref{SE01}, $\Delta=1, b=-1, n=2$ and $-1$ is a quadratic nonresidue of $\mathbb{F}_p$. So $|T_3|=p+1$, and then
\begin{equation}\label{CM022}
M_3=(p+1)(q-1)+1.
\end{equation}
Substituting to equation (\ref{CM02}), the third statement of the lemma is obtained. \qed
% In this case $p=3$, $x=\pm 1,y=\pm 1$. So $T'=4$, and $M_3=4(q-1)+1$.  \qed
\end{proof}

Corresponding to Lemma \ref{FL002} and Corollary \ref{RQ02}, we introduce the following notations for convenience.
Let
\begin{equation}\label{0001}
N_{\varepsilon,j}=\left\{(\alpha,\beta,\gamma)\in \mathbb{F}_q^3\backslash\{(0,0,0)\}|S(\alpha,\beta,\gamma)=\varepsilon p^{{m+j}\over 2}\right\}
\end{equation}
where $\varepsilon=\pm 1$ and $j=1,3$. Also, denote $n_{\varepsilon,j}=|N_{\varepsilon,j}|$ for $j=1,3$.
And
\begin{equation}\label{0002}
N_{\varepsilon,j}=\left\{(\alpha,\beta,\gamma)\in \mathbb{F}_q^3\backslash\{(0,0,0)\}|S(\alpha,\beta,\gamma)=\varepsilon i p^{{m+j}\over 2}\right\}
\end{equation}
for $j=0,2,4$, where $i$ is the imaginary unit.
 By Lemma \ref{RO01}, set
 \[
 n_{j}=n_{\varepsilon,j}=|N_{\varepsilon,j}|
 \]
  for $j=0,2,4$, since $m-j$ is odd.
 Using the above notations, Lemma \ref{ES001} can be restated as follows.

\begin{lemma}\label{ES002}
Let $p$ be an odd prime satisfying $p\equiv 3 \ \mbox{mod}\ 4$, and $q=p^m$ where $m$ is an odd integer with property $3\nmid m$.
\[
\begin{array}{ll}
2(n_0+n_2+n_4)+n_{-1,1}+n_{1,1}+n_{-1,3}+n_{1,3} & =p^{3m}-1\\
   n_{1,1}-n_{-1,1}
 +p(n_{1,3}-n_{-1,3})  &=p^{{m-1}\over 2}\left(p^{2m}-1\right) \\
 -2\left(n_0+p^2n_2+p^4n_4\right)
  +p (n_{1,1}+n_{-1,1})
+p^{3 }(n_{1,3}+n_{-1,3}) &=p^{m} \left(p^m-1\right)\\
    n_{1,1}-n_{-1,1}
 +p^{3}(n_{1,3}-n_{-1,3})  &= (p+1)p^{{3(m-1)}\over 2}\left(p^{m}-1\right)
\end{array}
\]
\end{lemma}

\begin{proof}
Substituting the symbols of (\ref{0001}) and (\ref{0002}) to Lemma \ref{ES001}, we have the following four equations
 \[
\begin{array}{lll}
&&2(n_0+n_2+n_4)+n_{-1,1}+n_{1,1}+n_{-1,3}+n_{1,3} \\
 &=&p^{3m}-1\\
\\
&&\sum\limits_{\alpha,\beta,\gamma\in \mathbb{F}_q} S(\alpha,\beta,\gamma)\\
&=&i{p^{m\over 2}}(n_{1,0}-n_{-1,0})+p^{{m+1}\over 2}(n_{1,1}-n_{-1,1})\\
&&+i{p^{{m+2}\over 2}}(n_{1,2}-n_{-1,2})+p^{{m+3}\over 2}(n_{1,3}-n_{-1,3})+i{p^{{m+4}\over 2}}(n_{1,4}-n_{-1,4})+p^m\\
&=&p^{{m+1}\over 2}(n_{1,1}-n_{-1,1})+p^{{m+3}\over 2}(n_{1,3}-n_{-1,3})+p^m\\
&=&p^{3m}
\end{array}
\]
 \[
\begin{array}{lll}
&&\sum\limits_{\alpha,\beta,\gamma\in \mathbb{F}_q} S(\alpha,\beta,\gamma)^2\\
&=&-{p^{m}}(n_{1,0}+n_{-1,0})+p^{{m+1} }(n_{1,1}+n_{-1,1})\\
&&-{p^{{m+2} }}(n_{1,2}+n_{-1,2})+p^{{m+3} }(n_{1,3}+n_{-1,3})- {p^{{m+4}}}(n_{1,4}+n_{-1,4})+p^{2m}\\
&=&p^{3m}\\
\\
 &&\sum\limits_{\alpha,\beta,\gamma\in \mathbb{F}_q} S(\alpha,\beta,\gamma)^3\\
&=&-i{p^{{3m}\over 2}}(n_{1,0}-n_{-1,0})+p^{{3(m+1)}\over 2}(n_{1,1}-n_{-1,1})\\
&&-i{p^{{3(m+2)}\over 2}}(n_{1,2}-n_{-1,2})+p^{{3(m+3)}\over 2}(n_{1,3}-n_{-1,3})-i{p^{{3(m+4)}\over 2}}(n_{1,4}-n_{-1,4})+p^{3m}\\
&=& p^{{3(m+1)}\over 2}(n_{1,1}-n_{-1,1})+p^{{3(m+3)}\over 2}(n_{1,3}-n_{-1,3})++p^{3m}\\
&=&\left((p+1)\left(p^m-1\right)+1\right)p^{3m}
\end{array}
\]
 where the first one comes from the fact that there are $p^{3m}-1$ elements in the set $\mathbb{F}_q^3\backslash\{(0,0,0)\}$. Also, note that $S(\alpha,\beta,\gamma)=p^m$ when $\alpha=\beta=\gamma=0$.

Using $n_{j}=n_{\varepsilon,j}=|N_{\varepsilon,j}|$ for $j=0,2,4$, the result is obtained by simplification. \qed
\end{proof}

\subsection{The fourth moment of $S(\alpha,\beta,\gamma)$} \label{Sec3.2}

For the fourth moment of $S(\alpha,\beta,\gamma)$ in the particular case of $p=3$, we calculate the
number of solutions of the following equation system
\begin{equation}\label{E02}
\left\{
\begin{array}{ll}
x^2+y^2+z^2+1 &=0\\
x^{p+1}+y^{p+1}+z^{p+1}+1 &=0\\
x^{p^2+1}+y^{p^2+1}+z^{p^2+1}+1 &=0
\end{array}
\right.
\end{equation}
in Lemma \ref{CM001}, which is denoted by $T_4$.

\begin{lemma}\label{CM001}
Let $p=3$ and $q=p^m$, then
\[
T_4=4\left(2p^m-3\right).
\]
\end{lemma}

\begin{proof}
The following process is composed of three parts: Part I is to find the values of the possible solutions $(x_0,y_0,z_0)$; Part II is to verify that they are actually a solution of equation (\ref{E02}); Part III is to find the number of the solutions.

Part I: To study equation system (\ref{E02}), this part tries to get the formula (\ref{E03001}) by using (\ref{E03},\ref{E0301},\ref{E0302}), and then obtain the solution cases (\ref{E0303},\ref{E0304}). Consider the first two equations
\begin{equation}\label{E03}
\left\{
\begin{array}{ll}
x^2+y^2+z^2+1 &=0\\
x^{p+1}+y^{p+1}+z^{p+1}+1 &=0,
\end{array}
\right.
\end{equation}
we find that
\begin{equation}\label{E0301}
z^{2}=-(x^{2}+y^2+1).
\end{equation}
Substituting (\ref{E0301}) to the second one of (\ref{E03})
\[
\begin{array}{ll}
x_0^{p+1}+y_0^{p+1}+z_0^{p+1}+1&= x^4+y^4+\left(-(x^{2}+y^2+1)\right)^2+1\\
&=x^4+y^4+x^4+y^4+1+2x^2+2y^2+2x^2y^2+1\\
&=2x^4+2y^4+2x^2y^2+2x^2+2y^2+2\\
&=0,
\end{array}
\]
that is
\begin{equation}\label{E0302}
x^4+y^4+x^2y^2+x^2+y^2+1=0.
\end{equation}
For equation (\ref{E0302}), set $x=x_0$ and consider $y$ as the variable to be determined, then
\[
x_0^4+y^4+x_0^2y^2+x_0^2+y^2+1=y^4+\left(x_0^2+1\right)y^2+\left(x_0^4+x_0^2+1\right).
\]
Set $y'=y^2$, the above equation becomes
\begin{equation}\label{E03001}
y'^2+(x_0^2+1)y'+(x_0^4+x_0^2+1)
\end{equation}
which is a quadratic polynomial about variable $y'$ over the finite field $\mathbb{F}_q$.
Set $b=x_0^2+1$ and $c=x_0^4+x_0^2+1$, then
\[
\begin{array}{ll}
\Delta &=b^2-4c\\
&=(x_0^2+1)^2-4(x_0^4+x_0^2+1)\\
&=x_0^4+2x_0^2+1-(x_0^4+x_0^2+1)\\
&=x_0^2.
\end{array}
\]
Note that the characteristic of the finite field $\mathbb{F}_q$ is $3$ and the elements in $\mathbb{F}_3$ is $0,1,2=-1$.
Corresponding to the solutions of equation (\ref{E0302})
\[
\begin{array}{ll}
y'&={{-b\pm \sqrt{\Delta}}\over 2}\\
 &=b\pm \sqrt{\Delta}\\
 &=x_0^2+1\pm x_0\\
 &=(x_0\pm 1)^2
\end{array}
\]
where we have used the fact that $x_0^2+x_0+1=x_0^2-2x_0+1=(x_0-1)^2$ and $x_0^2-x_0+1=x_0^2+2x_0+1=(x_0+1)^2$.
That is to say
$y'=y^2=(x_0\pm 1)^2,$
so
\begin{equation}\label{E0303}
y=x_0\pm 1 \ \ \ \mbox{or} \ \ \ y=-x_0\pm 1.
\end{equation}

Let's consider the first case of (\ref{E0303}) where $y=y_0=x_0\pm 1$. By the first equation of (\ref{E03})
\[
\begin{array}{ll}
x_0^2+y_0^2+z^2+1&=x_0^2+(x_0\pm 1)^2+z^2+1\\
&=x_0^2+x_0^2\pm 2x_0 +1+z^2+1\\
&=2\left(x_0^2\pm x_0 +1\right)+z^2\\
&=-(x_0\mp 1)^2+z^2 \\
&=0,
\end{array}
\]
i.e., $
z^2=(x_0\mp 1)^2 $
and
\begin{equation}\label{E0304}
z=z_0=x_0\mp 1 \ \ \ \mbox{or} \ \ \ z=z_0=-x_0\pm 1
\end{equation}
note that the symbols $+,-$ in (\ref{E0304}) are taken with respect to the symbols $\pm$ of (\ref{E0303}) when we set $y=y_0$.

  Part II: For (\ref{E0303}) and (\ref{E0304}), this part considers two first cases. Now, for the possible solutions $x=x_0, y=y_0=x_0\pm 1$ and $z=z_0=x_0\mp 1$, substituting to the second equation of (\ref{E02}) we verify that
\[
\begin{array}{ll}
x_0^{p+1}+y_0^{p+1}+z_0^{p+1}+1&=x_0^4+(x_0\pm 1)^4+(x_0\mp 1)^4+1 \\
 &=x_0^4+\left(x_0^3\pm 1\right)(x_0\pm 1)+\left(x_0^3\mp 1\right)(x_0\mp 1)+1\\
 &=x_0^4+x_0^4\pm x_0^3\pm x_0 + 1 +x_0^4 \mp x_0^3\mp x_0+1+1\\
 &=3\left(x_0^4+1\right)\\
 &=0
\end{array}
\]
which satisfies the second equation of (\ref{E02}).

Substituting the above values of $x_0,y_0$ and $z_0$ to the third equation of (\ref{E02}),
\[
\begin{array}{ll}
x_0^{p^2+1}+y_0^{p^2+1}+z_0^{p^2+1}+1&=x_0^{10}+ (x_0\pm 1)^{10}+(x_0\mp 1)^{10}+1\\
&=x_0^{10}+\left(x_0^9\pm 1\right)(x_0\pm 1)+\left(x_0^9\mp 1\right)(x_0\mp 1)+1\\
&=x_0^{10}+x_0^{10}\pm x_0^9 \pm x_0 +1 + x_0^{10}\mp x_0^{9}\mp x_0 +1 +1 \\
&=3\left(x_0^{10}+1\right)\\
&=0
\end{array}
\]
which implies that the third equation of (\ref{E02}) is also satisfied. Therefore the possible solutions $(x_0,y_0,z_0)$ satisfy system (\ref{E02}).

As to the first case of (\ref{E0303}) and the second case of (\ref{E0304}), since $a^2=(-a)^2$ for any $a\in \mathbb{F}_q$, it can be checked that $x_0,y_0=x_0\pm 1,z_0=-x_0\pm 1$ also satisfy (\ref{E02}).
For other cases, similar results can be obtained. %the second case of (\ref{E0303}), find that $x_0,y_0=-x_0\pm 1$ and $z_0=-x_0\mp 1$ or $z_0=x_0\pm 1$ satisfy the equation system (\ref{E02}).

  Part III: For the values in (\ref{E0303}), easy to see that $x_0+1\not= x_0-1$ and $-x_0+1\not=-x_0-1$.
  \begin{enumerate}
\renewcommand{\labelenumi}{$($\mbox{\roman{enumi}}$)$}
\item If $x_0+1=-x_0+1$ then $x_0=0$;
\item If $x_0+1=-x_0-1$ then $x_0=-1$;
\item If $x_0-1=-x_0+1$ then $x_0=1$;
\item If $x_0-1=-x_0-1$ then $x_0=0$.
\end{enumerate}
Therefore the possible values of $x_0$ which can produce same $y_0$ in (\ref{E0303}), are $0,1,-1$. And all other values of $x_0\in \mathbb{F}_q$ will reduce to different values of $y_0$. The same situation occurs for the cases of $z_0$. Let's consider these particular values of $x_0$. %Based on the above analysis, there are three cases to be considered
\begin{enumerate}
\renewcommand{\labelenumi}{$($\mbox{\roman{enumi}}$)$}
\item Suppose $x_0=0$.
\begin{itemize}
\renewcommand{\labelitemi}{\labelitemiii}
\item
In equation (\ref{E0303}), if $y_0=x_0+1=1$, then $z_0=x_0-1=-1$ or $-x_0+1=1$ in (\ref{E0304}).
\item If $y_0=x_0-1=-1$, then $z_0=x_0+1=1$ or $-x_0-1=-1$.
\item If $y_0=-x_0+1=1$, then $z_0=-x_0-1=-1$ or $z_0=x_0+1=1$.
\item If $y_0=-x_0-1=-1$, then $z_0=-x_0+1=1$ or $z_0=x_0-1=-1$.
\end{itemize}
 So there are $4$ solutions of $(x_0,y_0,z_0)$ when $x_0=0$.
\item For $x_0=1$ and $x_0=-1$, there are also $4$ solutions respectively.
\end{enumerate}
  Altogether there are $12$ solutions for $x_0=0,1$ and $-1$.

For each $x_0\not\in \{0,1,-1\}$, there are $4$ cases of $y_0$ in (\ref{E0303}), and for each selected $y_0$, there are two choices for $z_0$ in (\ref{E0304}), which leads to $4\times 2=8$ solutions of $(x_0,y_0,z_0)$.
So the number of solutions of (\ref{E02}) is $T_4=12+\left(p^m-3\right)\cdot 8=4\left(2p^m-3\right)$.   \qed

\end{proof}

\begin{lemma}\label{CM01}
Let $p=3$ and $q=p^m$ where $m$ is an odd integer satisfying $3\nmid m$. The number of solutions of the following equation system
\begin{equation}\label{E01}
\left\{
\begin{array}{ll}
x^2+y^2+z^2+w^2 &=0\\
x^{p+1}+y^{p+1}+z^{p+1}+w^{p+1} &=0\\
x^{p^2+1}+y^{p^2+1}+z^{p^2+1}+w^{p^2+1} &=0\\
\end{array}
\right.
\end{equation}
is $M_4=8\left(p^m-1\right)^2+1$.
\end{lemma}

\begin{proof}
For $w\not=0$, divide the three equations in (\ref{E01}) by $w^2, w^{p+1}$ and $w^{p^2+1}$ respectively, then equation system (\ref{E02}) is obtained, and the number of solutions of which is $T_4$ by Lemma \ref{CM001}. For $w=0$, the number of solutions of (\ref{E01}) is $M_3$ (equation \ref{CM022}) which assume that $m$ is an odd integer satisfying $3\nmid m$. Altogether, the number of solutions of equation system (\ref{E01}) is $M_4=\left(p^m-1\right)T_4+M_3=8\left(p^m-1\right)^2+1$.    \qed
\end{proof}

%Now, we have
Applying Lemma \ref{CM01}, the following result about the fourth moment of the exponential sum $S(\alpha,\beta,\gamma)$ can be obtained.
\begin{lemma}\label{S04}
Let $p=3$ and $q=p^m$ where $m$ is an odd integer satisfying $3\nmid m$. Then
\[
\sum\limits_{\alpha,\beta,\gamma\in \mathbb{F}_q} S(\alpha,\beta,\gamma)^4=M_4\cdot p^{3m}=\left(8(p^m-1)^2+1\right)p^{3m}.
\]
\end{lemma}

Using the symbols of (\ref{0001}) and (\ref{0002}), Lemma \ref{S04} can be rewritten as the following corollary.
\begin{corollary}\label{C02}
Let $p=3$ and $q=p^m$ where $m$ is an odd integer satisfying $3\nmid m$. Then
\[
2n_0+p^2(n_{-1,1}+n_{1,1})+p^4\cdot 2n_2+p^6(n_{-1,3}+n_{1,3})+p^8\cdot 2n_4=\left(8(p^m-1)^2-p^m+1\right)p^m.
\]
\end{corollary}

\subsection{Association schemes}\label{Sec3.3}

The following introduction is about skew-symmetric matrix (Lemma \ref{MS04}) and symmetric matrix (Proposition \ref{EF002}), and a relevant discussion \cite{LL001}.
 In fact, they correspond to two association schemes,
the fundamental properties of which are referred to \cite{CV,D02,MS001,S01}. Note that there exists a one-to-one correspondence between the set of alternating bilinear forms and the set of skew-symmetric matrices, and a one-to-one correspondence from the set of quadratic forms to the set of symmetric matrices.

A skew-symmetric matrix $B = [b_{i,j}]$ of order $m$ is a matrix which satisfies
\[
b_{i,i}=0, \ \ \ b_{i,j}+b_{j,i} = 0,
\]
and it has even rank.
Let $Y_m=Y(m,p)$ denote the set of skew-symmetric matrices of order $m$ over $\mathbb{F}_p$. It can be checked that $Y_m$ is an
${m(m-1)\over 2}$-dimensional vector space over $\mathbb{F}_p$.

Set $n=\lfloor m/2\rfloor$.
For $k=0,1,\ldots,n$, the partition $R' =\left\{R_0',R_1',\ldots,R_n' \right\}$ of $Y_m^2=Y_m \times Y_m$ is defined by
\[
R_k' =\{(A,B)\in Y_M^2|\mbox{rank}(A-B)=2k\}.
\]
Then $(Y_m,R')$ is an association scheme with $n$ classes \cite{DG01}.
The distance distribution of a nonempty subset $Y$ of $Y_m$ in the scheme $(Y_m,R')$ is the $(n+1)$-tuple $\mathbf{a} = (a_0,a_1,\ldots,a_n)$ of rational numbers $a_i$, where $|Y|a_i = |Y^2\bigcap R_i'|$. Easy to see that
\[
a_0 = 1, \ \ \  \mbox{and} \ \ \ a_0+a_1+\cdots+a_n=|Y|.
\]

\noindent An $(m,d)$-set $Y$ is a subset of $Y_m$ satisfying that
\[
\mbox{rank}(A-B)\geq 2d, \ \ \ \forall A,B \in Y, \ \ \ A\not= B,
\]
where $1\leq d\leq n$. In other words,
\[
a_1=a_2=\cdots =a_{d-1} = 0.
\]

For a real number $b\not= 1$ and all nonnegative integers $k$, denote the Gaussian binomial coefficients with basis $b$ by
$\left[ \begin{array}{c}
x\\
k
\end{array}
\right]_b$:

\[\left[ \begin{array}{c}
x\\
0
\end{array}
\right]_b =1 , \ \ \ \left[ \begin{array}{c}
x\\
k
\end{array}
\right]_b = \prod\limits_{i=0}^{k-1}{{(b^x-b^i)}/{(b^k-b^i)}},  \ \ \quad k=1,2,\ldots.
\]
 Set $b=p^2$ and $c = p^{m(m-1)/{2n}}$. The following result is about $(m,d)$-set \cite{DG01}.
\begin{lemma}(Theorem 4., \cite{DG01})\label{MS04}
\begin{enumerate}
\renewcommand{\labelenumi}{$($\roman{enumi}$)$}
\item
For any $(m,d)$-set $Y$, we have (the Singleton bound)
\[
|Y| \leq c^{n-d+1}.
\]
\item
In case of equality, the distance distribution of $Y$ is uniquely determined by
\[
a_{n-i} = \sum\limits_{j=i}^{n-d}(-1)^{j-i}
{b^{\left( \begin{array}{c}
j-i\\
2
\end{array}
\right)}
{\left[ \begin{array}{c}
j\\
i
\end{array}\right]_b}
{\left[ \begin{array}{c}
n \\
j
\end{array}\right]_b}(c^{n-d+1-j}-1), }
\]
for $i=0,1,\ldots,n-d$. Here, $\left(\begin{array}{c}j-i\\2\end{array}\right)$ represents the general binomial coefficient.
\end{enumerate}
 \end{lemma}

The set of symmetric matrices $X_m=X(m,p)$ forms a vector space of dimension $m(m+1)/2$ over $\mathbb{F}_p$.
Let
\[
R = \{R_i\},  \ \ \ i=0,1,2,\ldots, \lfloor {{m+1}\over 2}\rfloor,
\]
be the set of symmetric relations $R_i$ on $X_m$ defined by
\[
R_i = \{ (A,B)|A,B \in X_m, \mbox{rank}(A-B)=2i-1 \ \mbox{or} \ 2i   \}.
\]
Comparing to the association scheme $(Y_m,R')$ of skew-symmetric matrices, the following two lemmas are about the scheme $(X_m,R)$ \cite{E001}.
\begin{lemma} (Theorem 1., \cite{E001})
$(X_m,R)$ forms an association scheme of class $\lfloor(m+1)/2\rfloor$.
\end{lemma}

\begin{lemma}(Theorem 2., \cite{E001})\label{PQS}
All the parameters (and consequently all the eigenvalues) of the two association schemes $(X_m,R)$ and $(Y_{m+1},R')$ of class $\lfloor(m+1)/2\rfloor$ are exactly the same.
\end{lemma}
Distance distribution of a nonempty subset $X$ of $X_m$ % and $P$-transform of $(X_m,R)$
 can be defined as that of $(Y_m,R')$.
A set $X\subset X_m$ is called an $(m,d)$-set if it satisfies
\[
\mbox{rank}(A-B)\geq 2d-1, \ \ \ \forall A,B \in X, \ \ \ A\not= B,
\]
where $1\leq d \leq \lfloor(m+1)/2\rfloor$.

Similar to Lemma \ref{MS04}, Proposition \ref{EF002} is about the association scheme $(X_m,R)$.
\begin{proposition}\label{EF002}
 Set $b=p^2, c=p^{m(m+1)/{2n}}, n = \lfloor(m+1)/2\rfloor$.
\begin{enumerate}
\renewcommand{\labelenumi}{$($\roman{enumi}$)$}
\item
For any $(m,d)$-set $X\subset X_m$, we have (the Singleton bound)
\[
|X| \leq c^{n-d+1}.
\]
\item
In case of equality, the distance distribution of $X$ is uniquely determined by
\[
a_{n-i} = \sum\limits_{j=i}^{n-d}(-1)^{j-i}
{b^{\left( \begin{array}{c}
j-i\\
2
\end{array}
\right)}
{\left[ \begin{array}{c}
j\\
i
\end{array}\right]_b}
{\left[ \begin{array}{c}
n \\
j
\end{array}\right]_b}(c^{n-d+1-j}-1), }
\]
for $i=0,1,\ldots,n-d.$ Here, $\left(\begin{array}{c}j-i\\2\end{array}\right)$ represents the general binomial coefficient.
\end{enumerate}
\end{proposition}

\begin{proof}
 Refer to the proof of Lemma \ref{MS04} \cite{DG01}.  \qed
\end{proof}

\subsection{Main results}\label{Sec3.4}

This subsection contains the main results about the cyclic code $\mathcal{C}_1$. Before going on, Lemma \ref{MS02} is for the distance distribution of relevant subset $X$ of $X_m $.
With notations as in Proposition \ref{EF002} for $m=2t+1$, we have $n=t+1$ and $c=p^m$. In fact, Lemma \ref{ES002}, Corollary \ref{C02} and Corollary \ref{CQ02} result in Corollary \ref{CQ002}, and then lead to Theorem \ref{C002}.
\begin{lemma}\label{MS02}
Let $m=2t+1, t\geq 2$ and $q=p^m-1$ where $p$ is an odd prime.
Let $X$ denote the set of quadratic forms corresponding to $f_2'(x) =\alpha_0 x^2+\alpha_1x^{p+1}+\alpha_2x^{p^2+1}, \alpha_0,\alpha_1,\alpha_2 \in \mathbb{F}_q$ (equation \ref{QF01}).
  Then $X$ is an $(m,t-1)$-set, and $a_{n-i}$ is the number of quadratic forms in $X$ with rank $2(n-i)$ or $2(n-i)-1$ for $i=0,1,2$.
\end{lemma}

\begin{proof}
By Corollary \ref{RQ02}, the ranks of the quadratic forms contained in $X$ satisfy $r\geq m-4=2t+1-4=2(t-1)-1$. So $X$ is an $(m,t-1)$-set.
For $i=0,1,2$, by definition
\begin{equation}\label{QF0001}
|X|a_{n-i}=|X^2\bigcap R_{n-i}|.
 \end{equation}
 Let $X_i'$ denote the set of quadratic forms of rank $2(n-i)$ or $2(n-i)-1$ in $X$. Easy to find that the sum of two quadratic forms from $X$ also lies in $X$. For given $A\in X$, we find that $(A,C)=(A,A+B)\in R_{n-i}$ for any $B\in X_i'$, here $C=A+B$. And no element else in $X$ satisfies this property. So
 \begin{equation}\label{QF0002}
 |X^2\bigcap R_{n-i}|=|X||X_i'|,
 \end{equation}
  and $a_{n-i}=|X_i'|$ by comparing (\ref{QF0001}) with (\ref{QF0002}). \qed
\end{proof}

\begin{corollary}\label{CQ02}
Let $m=2t+1, t\geq 2$ and $q=p^m-1$ where $p$ is an odd prime.
The notations in (\ref{0001}) and (\ref{0002}) satisfy
\[
2n_0=a_{n}, n_{-1,1}+n_{1,1}+2n_2=a_{n-1},\ \mbox{and}\ \ n_{-1,3}+n_{1,3}+2n_4=a_{n-2}
\]
where $a_{n-i}$ is as defined in Proposition \ref{EF002} for $i=0,1,2$.
\end{corollary}

\begin{proof}
With notations as stated in Lemma \ref{MS02}, $X$ is an $(m,d)$-set with $d=t-1$ and $n-d=(t+1)-(t-1)=2$. By the first statement of Proposition \ref{EF002}, $|X|\leq c^{n-d+1}=p^{3m}$. Since the size of $X$ is $p^{3m}$, applying the second statement of Proposition \ref{EF002}, the result is obtained from Lemma \ref{MS02}.  \qed
\end{proof}

\begin{corollary}\label{CQ002}
Let $p=3$ and $q=p^m$ where $m>1$ is an odd integer satisfying $3\nmid m$.
The notations in (\ref{0001}) and (\ref{0002}) satisfy
\[
\begin{array}{ll}
2n_0&=a_{n}\\
2n_4 &={{B-A}\over {p^2(p^2-1)}} \\
2n_2 &={{p^2A-B}\over {p^2-1}} \\
2n_{1,3}&=p^{{m-3}\over 2}{{(p^m-1)(p^{m-1}-1)}\over {p^2-1}}+p^{3m}-1-a_n-a_{n-1}-2n_4 \\
2n_{-1,3}&=p^{3m}-1-a_n-a_{n-1}-2n_4-p^{{m-3}\over 2}{{(p^m-1)(p^{m-1}-1)}\over {p^2-1}} \\
2n_{1,1}&={{p^{m+2}-4p^{m-1}+p^2}\over {p^2-1}}\cdot p^{{m-1}\over 2}(p^m-1)+a_{n-1}-2n_2 \\
2n_{-1,1}&=a_{n-1}-2n_2-{{p^{m+2}-4p^{m-1}+p^2}\over {p^2-1}}\cdot p^{{m-1}\over 2}(p^m-1)
\end{array}
\]
where
\[
\begin{array}{ll}
A&={{p^{3m+2}-p^{m-1}(p^m-1)-p^2}\over {p+1}}-{{a_n(p^2-p+1)}\over p}-a_{n-1}(p-1), \\
B&={{(p^m-1)(8p^m-9)p^{m-2}+p^4-p^{3m+4}}\over {p^2-1}}+a_n\cdot {{p^4+p^2+1}\over p^2}+a_{n-1}\cdot (p^2+1)\\
 &={{(p^m-1)(8p^m-9)p^{m-2}-p^{3m}+1}\over {p^2-1}}-(p^2+1)(p^{3m}-1)+a_n\cdot {{p^4+p^2+1}\over p^2}+a_{n-1}\cdot (p^2+1),
\end{array}
\]
and $a_{n-i}$ is as defined in Proposition \ref{EF002} for $i=0,1,2$.
\end{corollary}

\begin{proof}
From Lemma \ref{ES002}, Corollary \ref{C02} and Corollary \ref{CQ02}, there are the following equations
\[
\begin{array}{ll}
2n_0&=a_{n}\\
 n_{-1,1}+n_{1,1}+2n_2 &=a_{n-1}\\
2(n_0+n_2+n_4)+n_{-1,1}+n_{1,1}+n_{-1,3}+n_{1,3} & =p^{3m}-1\\
  n_{1,1}-n_{-1,1}
 +p(n_{1,3}-n_{-1,3})  &=p^{{m-1}\over 2}\left(p^{2m}-1\right) \\
 -2(n_0+p^2n_2+p^4n_4)
  +p (n_{1,1}+n_{-1,1})
+p^{3 }(n_{1,3}+n_{-1,3}) &=p^{m} \left(p^m-1\right)\\
   n_{1,1}-n_{-1,1}
 +p^{3}(n_{1,3}-n_{-1,3})  &= 4p^{{3(m-1)}\over 2}\left(p^{m}-1\right)\\
 2n_0+p^2(n_{-1,1}+n_{1,1})+p^4\cdot 2n_2+p^6(n_{-1,3}+n_{1,3})+p^8\cdot 2n_4 &=\left(8(p^m-1)^2-p^m+1\right)p^m.
\end{array}
\]
After simplification, we find that
\[
\begin{array}{ll}
2n_2+p^2\cdot 2n_4 &=A \\
2n_2+p^4\cdot 2n_4 &=B. \\
\end{array}
\]
From which $2n_2$ and $2n_4$ can be calculated, then all the other notations can be obtained. \qed

\end{proof}

\begin{remark}
It can be calculated in above paragraph that
\[
\begin{array}{ll}
a_n&=p^{3m}-1-{{p^{m+1}-1}\over {p^2-1}}(p^{2m}-1)+p^2\cdot {{p^{m+1}-1}\over {p^4-1}}\cdot {{p^{m-1}-1}\over {p^2-1}}\cdot (p^m-1); \\
a_{n-1}&= {{p^{m+1}-1}\over {p^2-1}}(p^{2m}-1)-(p^2+1)\cdot {{p^{m+1}-1}\over {p^4-1}}\cdot {{p^{m-1}-1}\over {p^2-1}}\cdot (p^m-1); \\
a_{n-2}&= {{p^{m+1}-1}\over {p^4-1}}\cdot {{p^{m-1}-1}\over {p^2-1}}\cdot (p^m-1).
\end{array}
\]

\end{remark}

\begin{lemma}
Let $p=3$ and $q=p^m$ where $m>1$ is an odd integer satisfying $3\nmid m$.
The number of solutions of the equation system
\begin{equation}\label{M05}
\left\{
\begin{array}{ll}
x^{2}+y^2+z^2+w^2+u^2 &=0\\
x^{p+1}+y^{p+1}+z^{p+1}+w^{p+1}+u^{p+1} &=0\\
x^{p^2+1}+y^{p^2+1}+z^{p^2+1}+w^{p^2+1}+u^{p^2+1} &=0,
\end{array}
\right.
\end{equation}
is
\[
p^{2m}+p^{{5-m}\over 2}\left(n_{1,1}-n_{-1,1}+p^5(n_{1,3}-n_{-1,3})\right)
\]
where $n_{1,1},n_{-1,1},n_{1,3}$ and $n_{-1,3}$ are those of Corollary \ref{CQ002}.
\end{lemma}

\begin{proof}
The following moment of exponential sum $S(\alpha,\beta,\gamma)$ satisfies
\[
\begin{array}{ll}
\sum\limits_{\alpha,\beta,\gamma \in \mathbb{F}_q}S(\alpha,\beta,\gamma)^5
&= p^{{5(m+1)}\over 2}(n_{1,1}-n_{-1,1})+p^{{5(m+3)}\over 2}(n_{1,3}-n_{-1,3})+{(p^{m})}^5\\
&=M_5\cdot p^{3m}
\end{array}
\]
where $M_5$ is the number of solutions of (\ref{M05}). Solving the above equation for $M_5$, the result is obtained. \qed
\end{proof}
Equation (\ref{M05}) considers the case for $5$ variables. Using the sixth moment of $S(\alpha,\beta,\gamma)$, the number of solutions can be calculated when there are $6$ variables, etc.

%\section{.....}

%.........

\begin{theorem}\label{C002}
Let $p=3$ and $q=p^m$ where $m>1$ is an odd integer satisfying $3\nmid m$.
 The cyclic code $\mathcal{C}_1$ with nonzeros $\pi^{-2},\pi^{-(p+1)}$ and $\pi^{-(p^2+1)}$ has five nonzero weights with distribution
\[
\begin{array}{ll}
A_{p^{m-1}(p-1)}&=2(n_0+n_2+n_4),\\
A_{p^{m-1}(p-1)-{{p-1}\over p}p^{{m+1}\over 2}}&=n_{1,1},\\
A_{p^{m-1}(p-1)+{{p-1}\over p}p^{{m+1}\over 2}}&=n_{-1,1},\\
A_{p^{m-1}(p-1)-{{p-1}\over p}p^{{m+3}\over 2}}&=n_{1,3},\\
A_{p^{m-1}(p-1)+{{p-1}\over p}p^{{m+3}\over 2}}&=n_{-1,3},
\end{array}
\]
where the notations on the right hand side are defined in Corollary \ref{CQ002}.
\end{theorem}

\begin{proof}
There is a relation between the weight of a codeword in $\mathcal{C}_1$ and the corresponding exponential sum (equation \ref{C01})
\[
w_H(c)=p^{m-1}(p-1)-{1\over p}R(\alpha,\beta,\gamma).
\]
%Corollary \ref{RQ02} states that there are five possible values of the ranks contained in the cyclic code $\mathcal{C}_1$, $m,m-1,m-2,m-3$ and $m-4$. By Lemma \ref{FL002},
 The possible values of corresponding exponential sums $S(\alpha,\beta,\gamma)$ are indicated in (\ref{0001}) and (\ref{0002}). By Lemma \ref{ES02}, the values of the sums $R(\alpha,\beta,\gamma)$ can be calculated, and then the corresponding codeword weights and distribution. \qed
\end{proof}

\begin{example}
Let $m=5, p=3$ and $q=p^m$. Let $\mathcal{C}_1$ be the cyclic code with nonzeros $\pi^{-(p^0+1)}=\pi^{-2}$, $\pi^{-(p^1+1)}=\pi^{-4}$ and $\pi^{-(p^2+1)}=\pi^{-10}$ where $\pi$ is a primitive element of the finite field $\mathbb{F}_q$. Using Matlab, $\mathcal{C}_1$ has five nonzero weights
\[
\begin{array}{l}
A_{162}= 9740258, A_{144}=2548260, A_{180}=2038608,\ \ \mbox{and}\\
  A_{108}=14520, A_{216}=7260,
\end{array}
\]
which verifies the result of Theorem \ref{C002}.
\end{example}

\section{The cyclic code $\mathcal{C}_2$}\label{Sec4}

Let $p=3$ and $q=p^m$ where $m\equiv 1 \ \mbox{mod} \ 4$ satisfying $3\nmid m$.
In this section, we study the cyclic code $\mathcal{C}_2$ with nonzeros $\pi^{-1},\pi^{-2},\pi^{-(p+1)}$ and $\pi^{-(p^2+1)}$ where $\pi$ is a primitive element of the finite field $\mathbb{F}_{p^m}$. Lemma \ref{C001} is for the calculation of the multiplicities of exponential sum $S'(\alpha,\beta,\gamma,\delta)$ which leads to the result of weight distribution in Theorem \ref{C022}.

For this, let's consider the possible values of the exponential sums $S'(\alpha,\beta,\gamma,\delta)$, $R'(\alpha,\beta,\gamma,\delta)$ (Remark \ref{R0022}) and their multiplicities. By Lemma \ref{FL002}
\[
S(\alpha,\beta,\gamma)=i^r\left({\Delta \over p}\right)p^{m-r/2}
\]
 where $r$ is the rank of the corresponding quadratic form, and $\Delta$ is defined in equation (\ref{D01}).
For clearness, we list the notations of (\ref{0001}) and (\ref{0002}) in Table \ref{TAB01}.
%\vspace{0.2cm}
%\hspace{-0.6cm}
 \begin{table}[htbp]
\renewcommand\thetable{\Roman{table}}
 \caption{ }\label{TAB01} 
%\caption{ }\label{T001}
 \begin{tabular}{|l|l|l|l|l|l|l|l|l|l|l|}
 %\hline
%\multicolumn{11}{|c|}{Table I} \\
\hline
\multicolumn{1}{|l|}{rank}&\multicolumn{2}{c}{m}&\multicolumn{2}{|c|}{m-1}&\multicolumn{2}{|c|}{m-2}&\multicolumn{2}{|c|}{m-3}&\multicolumn{2}{|c|}{m-4}\\
 \hline
 $S(\alpha,\beta,\gamma)$&$ip^{m\over 2}$&$-ip^{m\over 2}$ &$p^{{m+1}\over 2}$&$-p^{{m+1}\over 2}$&$ip^{{m+2}\over 2}$& $-ip^{{m+2}\over 2}$&$ p^{{m+3}\over 2}$& $-p^{{m+3}\over 2}$&$ip^{{m+4}\over 2}$& $-ip^{{m+4}\over 2}$ \\
 \hline
multiplicity &$n_0$& $n_0$ &$n_{1,1}$& $n_{-1,1}$&$n_2$& $n_2$&$n_{1,3}$& $n_{-1,3}$&$n_4$&$n_4$\\
 \hline
\end{tabular}
\end{table}

%\vspace{0.5cm}

\noindent For quadratic form $F(X)$ with corresponding symmetric matrix $H$ of rank $r$, depending on Table \ref{TAB01}, Lemma \ref{C001} considers the exponential sum
\begin{equation}\label{SE0003}
S'(\alpha,\beta,\gamma,\delta)=\sum\limits_{X\in \mathbb{F}_p^m}\zeta_p^{\mbox{Tr}\left(F(X)+AX^T\right)}
\end{equation}
where $A$ varies over $\mathbb{F}_p^m$.

\begin{lemma}\label{C001}
Let $p=3$ and $q=p^m$ where $m\equiv 1 \ \mbox{mod} \ 4$ satisfying $3\nmid m$.
The exponential sums $S'(\alpha,\beta,\gamma,\delta)$ and their multiplicities are listed in Table \ref{TAB02} where $(\alpha,\beta,\gamma)\in \mathbb{F}_q^3\backslash \{(0,0,0)\}$.
 \end{lemma}
 \begin{table}
 \renewcommand\thetable{\Roman{table}}
 \caption{ }\label{TAB02}
 \begin{tabular}{|l|l|l|l|}
%\hline
%\multicolumn{4}{|c|}{r=m}\\
\hline
\multicolumn{2}{|c|}{$\Delta=1$}&\multicolumn{2}{|c|}{$\Delta=2$}\\
\hline
 $S'(\alpha,\beta,\gamma,\delta)$ & multiplicity & $S'(\alpha,\beta,\gamma,\delta)$ & multiplicity \\
 \hline
0 &0 &0 & 0\\
\hline
$ip^{m\over 2}$& $n_0\cdot p^{m-1}$ &$ -ip^{m\over 2}$& $n_0\cdot p^{m-1}$ \\
\hline
$\zeta_pip^{m\over 2}$& $n_0\cdot (p^{m-1}-p^{{m-1}\over 2})$ &$ -\zeta_pip^{m\over 2}$& $n_0\cdot (p^{m-1}+p^{{m-1}\over 2})$ \\
\hline
$\zeta_p^2ip^{m\over 2}$& $n_0\cdot (p^{m-1}+p^{{m-1}\over 2})$ &$ -\zeta_p^2ip^{m\over 2}$& $n_0\cdot (p^{m-1}-p^{{m-1}\over 2})$ \\
\hline
\multicolumn{4}{|c|}{r=m-1}\\
\hline
\multicolumn{2}{|c|}{$\Delta=1$}&\multicolumn{2}{|c|}{$\Delta=2$}\\
\hline
 $S'(\alpha,\beta,\gamma,\delta)$ & multiplicity & $S'(\alpha,\beta,\gamma,\delta)$ & multiplicity \\
 \hline
0 &$n_{1,1}(p^m-p^{m-1})$ &0 & $n_{-1,1}(p^m-p^{m-1})$\\
\hline
$ p^{{m+1}\over 2}$& $n_{1,1}\cdot (p^{m-2}+(p-1)p^{{m-3}\over 2})$ &$ -p^{{m+1}\over 2}$& $n_{-1,1}\cdot (p^{m-2}-(p-1)p^{{m-3}\over 2})$\\
\hline
$\zeta_pp^{{m+1}\over 2}$& $n_{1,1}\cdot (p^{m-2}-p^{{m-3}\over 2})$ &$ -\zeta_pp^{{m+1}\over 2}$& $n_{-1,1}\cdot (p^{m-2}+p^{{m-3}\over 2})$ \\
\hline
$\zeta_p^2p^{{m+1}\over 2}$& $n_{1,1}\cdot (p^{m-2}-p^{{m-3}\over 2})$ &$ -\zeta_p^2p^{{m+1}\over 2}$& $n_{-1,1}\cdot (p^{m-2}+p^{{m-3}\over 2})$ \\
\hline
\multicolumn{4}{|c|}{r=m-2}\\
\hline
\multicolumn{2}{|c|}{$\Delta=1$}&\multicolumn{2}{|c|}{$\Delta=2$}\\
\hline
 $S'(\alpha,\beta,\gamma,\delta)$ & multiplicity & $S'(\alpha,\beta,\gamma,\delta)$ & multiplicity \\
 \hline
0 &$n_2(p^m-p^{m-2})$ &0 & $n_2(p^m-p^{m-2})$\\
\hline
$ -ip^{{m+2}\over 2}$& $n_2\cdot p^{m-3}$ &$  ip^{{m+2}\over 2}$& $n_2\cdot p^{m-3}$\\
\hline
$-i\zeta_pp^{{m+2}\over 2}$& $n_2\cdot (p^{m-3}+p^{{m-3}\over 2})$ &$ i\zeta_pp^{{m+2}\over 2}$& $n_2\cdot (p^{m-3}-p^{{m-3}\over 2})$ \\
\hline
$-i\zeta_p^2p^{{m+2}\over 2}$& $n_2\cdot (p^{m-3}-p^{{m-3}\over 2})$ &$ i\zeta_p^2p^{{m+2}\over 2}$& $n_2\cdot (p^{m-3}+p^{{m-3}\over 2})$ \\
\hline
\multicolumn{4}{|c|}{r=m-3}\\
\hline
\multicolumn{2}{|c|}{$\Delta=1$}&\multicolumn{2}{|c|}{$\Delta=2$}\\
\hline
 $S'(\alpha,\beta,\gamma,\delta)$ & multiplicity & $S'(\alpha,\beta,\gamma,\delta)$ & multiplicity \\
 \hline
0 &$n_{-1,3}(p^m-p^{m-3})$ &0 & $n_{1,3}(p^m-p^{m-3})$\\
\hline
$ -p^{{m+3}\over 2}$& $n_{-1,3}\cdot (p^{m-4}-(p-1)p^{{m-5}\over 2})$ &$  p^{{m+3}\over 2}$& $n_{ 1,3}\cdot (p^{m-4}+(p-1)p^{{m-5}\over 2})$\\
\hline
$-\zeta_pp^{{m+3}\over 2}$& $n_{-1,3}\cdot (p^{m-4}+p^{{m-5}\over 2})$ &$  \zeta_pp^{{m+3}\over 2}$& $n_{ 1,3}\cdot (p^{m-4}-p^{{m-5}\over 2})$ \\
\hline
$-\zeta_p^2p^{{m+3}\over 2}$& $n_{-1,3}\cdot (p^{m-4}+p^{{m-5}\over 2})$ &$  \zeta_p^2p^{{m+3}\over 2}$& $n_{ 1,3}\cdot (p^{m-4}-p^{{m-5}\over 2})$ \\
\hline
\multicolumn{4}{|c|}{r=m-4}\\
\hline
\multicolumn{2}{|c|}{$\Delta=1$}&\multicolumn{2}{|c|}{$\Delta=2$}\\
\hline
 $S'(\alpha,\beta,\gamma,\delta)$ & multiplicity & $S'(\alpha,\beta,\gamma,\delta)$ & multiplicity \\
 \hline
0 &$n_4(p^m-p^{m-4})$ &0 & $n_4(p^m-p^{m-4})$\\
\hline
$  ip^{{m+4}\over 2}$& $n_4\cdot p^{m-5}$ &$  -ip^{{m+4}\over 2}$& $n_4\cdot p^{m-5}$\\
\hline
$ i\zeta_pp^{{m+4}\over 2}$& $n_4\cdot (p^{m-5}-p^{{m-5}\over 2})$ &$ -i\zeta_pp^{{m+4}\over 2}$& $n_4\cdot (p^{m-5}+p^{{m-5}\over 2})$ \\
\hline
$ i\zeta_p^2p^{{m+4}\over 2}$& $n_4\cdot (p^{m-5}+p^{{m-5}\over 2})$ &$ -i\zeta_p^2p^{{m+4}\over 2}$& $n_4\cdot (p^{m-5}-p^{{m-5}\over 2})$ \\
\hline
\end{tabular}
\end{table}

\begin{proof}
 In the following, there are three parts for the proof: Part I is for general analysis and the case where (\ref{N001}) is not solvable; Part II gives the counting formula for the times that $S'(\alpha,\beta,\gamma,\delta)$ takes each possible value according to the rank of corresponding quadratic form; Part III presents the calculation of the counting formula.

Part I: Lemma \ref{FL002} implies that the value of exponential sum $S'(\alpha,\beta,\gamma,\delta)$ (\ref{SE0003}) is
\[
\zeta_p^cS(\alpha,\beta,\gamma)
\]
where $c={1\over 2}AB^T\in \mathbb{F}_p$. And the value is $0$ if
\begin{equation}\label{N001}
2YH+A=0
\end{equation}
 does not have a solution $Y=B \in \mathbb{F}_p^m$. For more details of (\ref{N001}), see follows.

Denote by $M\in \mbox{GL}_m(\mathbb{F}_p)$ such that
\[
MHM^T=H'=\mbox{diag}(h_1,h_2,\ldots,h_r,0,0,\ldots,0)
\]
where $h_i\in \mathbb{F}_p^*$. Let $Y'=2YM^{-1}$ and $A'=-AM^T$, then (\ref{N001}) is equivalent to
\begin{equation}\label{N01}
Y'H'=A'.
\end{equation}
Since $A$ varies over the elements of $\mathbb{F}_p^m$ and $M$ is nonsingular, $A'$ also varies over the elements of $\mathbb{F}_p^m$. And if solvable, it can be checked that
\begin{equation}\label{QE01}
\begin{array}{ll}
Y'A'^T&=2YM^{-1} \cdot \left(-AM^T\right)^T\\
 &=2YM^{-1} \cdot \left(-MA^T\right)=YA^T\\
 &=-\left({1\over 2}AB^T\right)\\
 &=-c
\end{array}
\end{equation}
where $B=-Y'M$ is a solution of equation (\ref{N001}).
Note that equation (\ref{QE01}) can also be written as follows
\begin{equation}\label{QE02}
\begin{array}{ll}
Y'A'^T &=Y'H'Y'^T \\
       & =h_1y_1'^2+h_2y_2'^2+\cdots +h_ry_r'^2\\
       &=-c
\end{array}
\end{equation}
where $y_1',y_2',\ldots,y_r'$ are the elements of the first $r$ coordinates of the vector $Y'\in \mathbb{F}_p^m$.

It can be checked that equation (\ref{N01}) is solvable only when the elements of the last $m-r$ coordinates of $A'$ are $0$, and in this case for any such $A'$ the number of solutions of equation (\ref{N01}) is $p^{m-r}$. The number of such vectors $A'$ is $p^{r}$. So $S'(\alpha,\beta,\gamma,\delta)$ is zero for
\begin{equation}\label{N02}
p^m-p^r
\end{equation}
 vectors $A'\in \mathbb{F}_p^m$ when equation (\ref{N01}) is not solvable. %The main results of this part are equations (\ref{QE02}) and (\ref{N02}).

Part II: In addition to (\ref{N02}), let's consider the case where $S'(\alpha,\beta,\gamma,\delta)$ is not zero. From Lemma \ref{FL002}, we find that
\[
S'(\alpha,\beta,\gamma,\delta)=\zeta_p^cS(\alpha,\beta,\gamma)
\]
with $c$ as explained above by equations (\ref{QE01}) and (\ref{QE02}).
There are two cases to be considered.
\begin{enumerate}
\renewcommand{\labelenumi}{$($\mbox{\roman{enumi}}$)$}
\item The rank $r$ of the corresponding quadratic form is even.
\begin{itemize}
\renewcommand{\labelitemi}{\labelitemiii}
\item For the case of $c=0$, equation (\ref{QE02}) becomes
\[
h_1y_1'^2+h_2y_2'^2+\cdots +h_ry_r'^2=0.
\]
By Lemma \ref{SE01} the number of solutions of the above equation is
\begin{equation}\label{SE0001}
p^{r-1}+(p-1)p^{{r-2}\over 2}\eta \left((-1)^{r\over 2}\Delta \right)
\end{equation}
where $\Delta =h_1'h_2'\cdots h_r'$ and $\eta$ is the quadratic character of $\mathbb{F}_p$. Note that, when $y_1',y_2',\ldots,y_r'$ are set, $A'$ is determined by equation (\ref{N01}) with $Y'=(y_1',y_2',\ldots,y_r',0,0,\ldots,0)$, and then $A$ can be calculated by $A'=-AM^T$.
\item For the cases of $c=1$ and $c=2$, the number of solutions is
\[
p^{r-1}-p^{{r-2}\over 2}\eta \left((-1)^{r\over 2}\Delta \right).
\]
\end{itemize}

\item The rank $r$ of the corresponding quadratic form is odd.
\begin{itemize}
\renewcommand{\labelitemi}{\labelitemiii}
\item  For the case of $c=0$, equation (\ref{QE02}) becomes
\[
h_1y_1'^2+h_2y_2'^2+\cdots +h_ry_r'^2=0.
\]
By Lemma \ref{SE02} the number of solutions of the above equation is
\[
p^{r-1}.
\]

\item For the case of $c=1$, the number of solutions is
\begin{equation}\label{SE0002}
p^{r-1}+p^{{r-1}\over 2}\eta \left((-1)^{{r-1}\over 2}(-1)\Delta \right).
\end{equation}

\item For the case of $c=2$, the number of solutions is
\[
p^{r-1}+p^{{r-1}\over 2}\eta \left((-1)^{{r-1}\over 2}\Delta \right).
\]
\end{itemize}
\end{enumerate}

Part III: This part considers the calculation of above counting formulas in the two cases when $r=m$ and $r=m-1$ for odd and even ranks respectively.
\begin{enumerate}
\renewcommand{\labelenumi}{$($\mbox{\roman{enumi}}$)$}
\item For the case of $r=m$, assume that $\Delta=1$, then $S(\alpha,\beta,\gamma)=ip^{m\over 2}$ and the number of such quadratic forms is $n_0$ (Table \ref{TAB01}).

\begin{itemize}
\renewcommand{\labelitemi}{\labelitemiii}
\item By equation (\ref{N02}), the number of times that $S'(\alpha,\beta,\gamma,\delta)=0$ is equal to
\[
n_0\left(p^m-p^m\right)=0.
\]
\item The number of times that $S'(\alpha,\beta,\gamma,\delta)=ip^{m\over 2} (c=0)$  is equal to
\[
n_0\cdot p^{m-1}.
\]

\item The number of times that $S'(\alpha,\beta,\gamma,\delta)=\zeta_p \cdot ip^{m\over 2} (c=1)$  is equal to
\[
n_0\cdot \left(p^{m-1}-p^{{m-1}\over 2}\right)
\]
where $\eta\left((-1)^{{m-1}\over 2}(-1)\Delta \right)=\eta(-1)=-1$ using \ref{SE0002}.

\item The number of times that $S'(\alpha,\beta,\gamma,\delta)=\zeta_p^2 \cdot ip^{m\over 2} (c=2)$  is equal to
\[
n_0\cdot \left(p^{m-1}+p^{{m-1}\over 2}\right)
\]
where $\eta\left((-1)^{{m-1}\over 2}\Delta \right)=1$.
\end{itemize}
Similarly, the case of $\Delta =2$ can be analyzed.

\item Now, let's consider the case of $r=m-1$, and assume that $\Delta =1$. In this case $S(\alpha,\beta,\gamma)=p^{{m+1}\over 2}$, and the number of such quadratic forms is $n_{1,1}$.
\begin{itemize}
 \renewcommand{\labelitemi}{\labelitemiii}
\item By equation (\ref{N02}), the number of times that $S'(\alpha,\beta,\gamma,\delta)=0$ is equal to
\[
n_{1,1}\cdot \left(p^m-p^{m-1}\right).
\]
\item The number of times that $S'(\alpha,\beta,\gamma,\delta)=p^{{m+1}\over 2} (c=0)$ is equal to
\[
n_{1,1}\cdot \left(p^{m-2}+(p-1)p^{{m-3}\over 2}\right),
\]
noting that $\eta((-1)^{{m-1}\over 2}\Delta)=1$ using \ref{SE0001}.
\item The number of times that $S'(\alpha,\beta,\gamma,\delta)=\zeta_pp^{{m+1}\over 2} (c=1)$ or $\zeta_p^2p^{{m+1}\over 2} (c=2)$ is equal to
\[
n_{1,1}\cdot \left(p^{m-2}-p^{{m-3}\over 2}\right).
\]
% \item The number of times that $S'(\alpha,\beta,\gamma,\delta)=\zeta_p^2p^{{m+1}\over 2} (c=2)$ is equal to
%\[
%n_{1,1}\cdot \left(p^{m-2}-p^{{m-3}\over 2}\right).
%\]
\end{itemize}
The case of $\Delta =2$ can also be investigated in this way.
\end{enumerate}
Using similar ideas on the other cases of $r$, the lemma is obtained. \qed
\end{proof}

The weight distribution of the cyclic code $\mathcal{C}_2$ is obtained in Theorem \ref{C022} by analyzing exponential sum $R'(\alpha,\beta,\gamma,\delta)$ from Lemma \ref{C001} and Remark \ref{R0022}.

\begin{theorem}\label{C022}
  Let $p=3$ and $q=p^m$ where $m\equiv 1 \ \mbox{mod} \ 4$ satisfying $3\nmid m$.
  The multiplicities of the exponential sums $R'(\alpha,\beta,\gamma,\delta)$ and the weight distribution of the cyclic code $\mathcal{C}_2$ are listed in Table \ref{TAB03}.
\end{theorem}
\begin{table}[htbp]
\renewcommand\thetable{\Roman{table}}
 \caption{ }\label{TAB03}
  \begin{tabular}{|l|l|l|}
 %\hline
 %\multicolumn{3}{|c|}{Table II} \\
\hline
$R'(\alpha,\beta,\gamma,\delta)$ &weight & multiplicity \\
\hline
 & & $2n_0p^{m-1}+(n_{-1,1}+n_{1,1})(p^m-p^{m-1})+2n_2(p^m-2p^{m-3})$\\
  \raisebox{1.6ex}[0pt]{0}& \raisebox{1.6ex}[0pt]{$p^{m-1}(p-1)$} & $+(n_{-1,3}+n_{1,3})(p^m-p^{m-3})+2n_4(p^m-2p^{m-5})+p^m-1$\\
  \hline
$-p^{{m+1}\over 2}$&$p^{m-1}(p-1)+p^{{m-1}\over 2}$& $2n_0(p^{m-1}-p^{{m-1}\over 2})+2n_{1,1}(p^{m-2}-p^{{m-3}\over 2})$\\
\hline
$ p^{{m+1}\over 2}$&$p^{m-1}(p-1)- p^{{m-1}\over 2}$& $2n_0(p^{m-1}+p^{{m-1}\over 2})+2n_{-1,1}(p^{m-2}+p^{{m-3}\over 2})$\\
\hline
$ p^{{m+3}\over 2}$&$ p^{m-1}(p-1)-p^{{m+1}\over 2}$ & $2n_2(p^{m-3}+p^{{m-3}\over 2})+2n_{-1,3}(p^{m-4}+p^{{m-5}\over 2})$\\
\hline
$-p^{{m+3}\over 2}$& $p^{m-1}(p-1)+p^{{m+1}\over 2}$& $2n_2(p^{m-3}-p^{{m-3}\over 2})+2n_{1,3}(p^{m-4}-p^{{m-5}\over 2})$\\
\hline
$(p-1)p^{{m+1}\over 2}$&$p^{m-1}(p-1)-(p-1)p^{{m-1}\over 2}$& $n_{1,1}(p^{m-2}+(p-1)p^{{m-3}\over 2})$\\
\hline
$-(p-1)p^{{m+1}\over 2}$&$p^{m-1}(p-1)+(p-1)p^{{m-1}\over 2}$&  $n_{-1,1}(p^{m-2}-(p-1)p^{{m-3}\over 2})$\\
\hline
$(p-1)p^{{m+3}\over 2}$& $p^{m-1}(p-1)-(p-1)p^{{m+1}\over 2}$ & $n_{1,3}(p^{m-4}+(p-1)p^{{m-5}\over 2})$\\
\hline
$-(p-1)p^{{m+3}\over 2}$&$p^{m-1}(p-1)+(p-1)p^{{m+1}\over 2}$ &$n_{-1,3}(p^{m-4}-(p-1)p^{{m-5}\over 2})$\\
\hline
$-p^{{m+5}\over 2}$& $p^{m-1}(p-1)+p^{{m+3}\over 2}$& $2n_4(p^{m-5}-p^{{m-5}\over 2})$\\
\hline
$ p^{{m+5}\over 2}$&$ p^{m-1}(p-1)-p^{{m+3}\over 2}$ & $2n_4(p^{m-5}+p^{{m-5}\over 2})$\\
\hline
\end{tabular}
\end{table}

%Note that $p^m-1$ in the third line of Table II, corresponds to the exponential sum $R'(\alpha,\beta,\gamma,\delta)$ for the cases where $\alpha,\beta,\gamma$ are all zero.

\begin{example}
Let $m=5, p=3$ and $q=p^m$. Let $\mathcal{C}_2$ be the cyclic code with nonzeros $\pi^{-1}, \pi^{-(p^0+1)}=\pi^{-2}$, $\pi^{-(p^1+1)}=\pi^{-4}$ and $\pi^{-(p^2+1)}=\pi^{-10}$ where $\pi$ is a primitive element of the finite field $\mathbb{F}_q$. Using Matlab, $\mathcal{C}_2$ has weight distribution
\[
\begin{array}{l}
A_{162}=1618713316, A_{171}=782825472, A_{153}=947952720,A_{135}=6853440, \\A_{189}=3455760,
A_{144}=84092580, A_{180}=42810768,  A_{108}=72600, A_{216}=7260, \\
A_{81}=484,
\end{array}
\]
which verifies the result of Theorem \ref{C022}.
\end{example}

\section{Conclusions}\label{Sec5}

%Even though the minimum distances of cyclic codes are in general not as good as some other linear codes of the same length and dimension, they are commonly applied in many areas. % because there are more efficient encoding and decoding algorithms.
% Using the polynomial representations, there are efficient hardware configurations for performing the encoding operations of cyclic codes. %To decode a linear block code, we can form a standard array, or we can use the reduced standard array using syndromes. For cyclic codes, i
%And it is possible to exploit the cyclic structures of the codes to further decrease the memory requirements instead of using the standard arrays in decoding cyclic codes.

In this paper, we describe the weight distributions of some cyclic codes: the codes $\mathcal{C}_1$ with nonzeros $\pi^{-2},\pi^{-(p+1)}$, $\pi^{-(p^2+1)}$ has five nonzero weights, and the code $\mathcal{C}_2$ with nonzeros $\pi^{-1},\pi^{-2},\pi^{-4},\pi^{-10}$ has ten nonzero weights respectively, where $p=3$, $q=p^m$ and $m$ is an odd integer satisfying $3\nmid m$. As can be expected, in general the weight formulas of cyclic codes are rather complicated.

%\section{Acknowledgment}

 %The authors are very grateful to the reviewers for their
%detailed comments and suggestions that much improved the
%quality of this paper.


\end{document}